%% file: Droop_control_UB_Rahmat_v6.tex
\documentclass[letter, 10pt, conference]{ieeeconf}      

\IEEEoverridecommandlockouts                              
\usepackage{amsmath}
\usepackage{amssymb}
\usepackage{graphicx}
\usepackage{xspace}
\usepackage{enumerate}
\usepackage{color}
\usepackage{epstopdf}
 
\newtheorem{theorem}{Theorem}[section]
\newtheorem{lemma}[theorem]{Lemma}
\newtheorem{corollary}[theorem]{Corollary}

\newtheorem{assu}[theorem]{Assumption}

\newtheorem{remark}[theorem]{Remark}

\title{Ultimate Boundedness of Droop Controlled Microgrids \\with
  Secondary Loops}

\author{Rahmat Heidari$^{1}$, Maria M.  Seron$^{1}$, Julio
 H. Braslavsky$^{2}$
 \thanks{$^{1}$ Priority Research Centre for Complex Dynamic Systems
   and Control (CDSC), School of Electrical Engineering and Computer
   Science, The University of Newcastle, Callaghan NSW 2308,
   Australia {\tt\scriptsize 
heidari.rahmat@gmail.com, 
     maria.seron@newcastle.edu.au}} %
 \thanks{$^{2}$ Australian Commonwealth Scientific and Industrial
   Research Organisation (CSIRO),  Energy Flagship,
   PO box 330, Newcastle, NSW 2300, Australia {\tt\scriptsize
     julio.braslavsky@csiro.au}} %
}

\input{texcommands_droop.tex}

\begin{document}
\maketitle
\thispagestyle{empty}
\pagestyle{empty}

\begin{abstract}
  In this paper we study theoretical properties of inverter-based
  microgrids controlled via primary and secondary loops.
  Stability of these microgrids has been the subject of a number of 
  recent studies.  Conventional approaches based on standard hierarchical
  control rely on time-scale separation between primary and secondary
  control loops to show local stability of equilibria. In this paper 
  we show that (i) frequency regulation can be ensured
  without assuming time-scale separation and, (ii) ultimate boundedness
  of the trajectories starting inside a region of the state space can be
  guaranteed under a condition on the inverters power injection
  errors. The trajectory ultimate bound can be computed by simple
  iterations of a nonlinear mapping and provides a certificate of the
  overall performance of the controlled microgrid.
\end{abstract}

\section{Introduction}
In the last decade, the need to mitigate the environmental impacts of
coal-fired electricity generation has stimulated a gradual transition
from large centralised energy grids towards small-scale distributed
generation (DG) of power \cite{Ustun2011}.  A common operating regime
for DG is to form microgrids before being connected to the main energy
grid.  A microgrid is a small-scale power system consisting of a
collection of DG units, loads and local storage, operating together
with energy management, control and protection devices and associated
software \cite{Lasseter01,PLT09}.

Control strategies are indispensable to provide stability in
microgrids \cite{PMM06}. Recently, hierarchical control for microgrids
has been proposed in order to standardise their operation and
functionalities \cite{GVMdVC11,Bidram2012}. In this hierarchical
approach, three main control levels are defined to manage voltage and
frequency stability and regulation, and power flow and economic
optimisation. In this paper we focus on the primary and secondary
control levels, which are the main parts of the automatic control
system for the microgrid.

The primary control level deals with the local control loops of the DG
sources.  Many of these sources generate either variable frequency AC
power or DC power, and are interfaced with an AC grid via power
electronic DC/AC inverters.  For inductive lines, inverters are
typically controlled to emulate the droop characteristic of
synchronous generators.  Conventionally, the frequency-active power
(or ``$\omega$-P'' ) droop control \cite{CDA93} is adopted as the
decentralised control strategy for the autonomous active power sharing
at primary level.  Because standard droop control is a purely
proportional control strategy, the secondary control level has the
task of compensating for frequency steady-state errors induced by the
primary control layer.  Although the secondary control level is
conventionally implemented in a centralised fashion, several recent
works have suggested distributed control implementations
\cite{SVG12,ASDJ12,SATS12}.

Stability and convergence properties of droop-controlled networks of
inverters and loads have recently been the focus of the detailed
analyses that highlight the dynamic properties of the power system
\cite{AiG13,ASDJ12,JWSP-FD-FB:12u}.
For example, in \cite{JWSP-FD-FB:12u}, the authors present a necessary 
and sufficient condition for the existence of a unique and locally
exponentially stable steady state equilibrium for a droop-controlled
network.  The paper also proposes a distributed secondary-control
scheme to dynamically regulate the network frequency to a nominal
value while maintaining proportional power sharing among the
inverters, and without assuming time-scale separation between primary
and secondary control loops.  This is in contrast with more
conventional analyses which rely on time-scale separation and do not
discuss stability properties beyond local results around equilibrium
points \cite{JWSP-FD-FB:12u}.

In this paper{\footnote{Preprint. Original version submitted to AuCC`14.}
 we analyse ultimate boundedness of the states of an
inverter-based purely inductive microgrid with decentralised droop
control and secondary control systems.
The network of our study is inherently decentralised as no communication between
neighbouring droop controllers is needed.
Our first contribution is a structured nonlinear model for a microgrid
with embedded primary and secondary control levels.  
By performing a suitable change of coordinates, we show how the
stability analysis for the controlled system is decoupled into a
linear system stability problem, and that of characterising
ultimate boundedness of the trajectories of a perturbed nonlinear
subsystem around steady-state solutions. Our second and main
contribution is then to establish stability properties of the original
nonlinear system by exploiting this model separation.  The linear
analysis shows that frequency regulation
is ensured without the need for time-scale separation.  For the
perturbed nonlinear subsystem, we show that ultimate boundedness of
the trajectories starting inside a region of the state space is
guaranteed under a condition on the power injection errors for the
inverters. The ultimate bounds for the trajectories can be computed by
iterating a well-specified nonlinear map, which provides key
certificates for the overall performance of the controlled microgrid.


\emph{Notation and Definitions:} Let $\mathbf{1}_n$ and $\mathbf{0}_n$
be the $n$-dimensional vectors of unit and zero entries.  Let $\I
\doteq \{1,2,\dots,n\}$
and $\J \doteq \{1,2,\dots,m\}$ be 
index sets of inverter buses and edges, respectively.  For a matrix
$M$, $M_{(i,:)}$, $M_{(:,j)}$, $M_{(i:j,:)}$ and $M_{(i,j)}$ denote
its $i$-th row, $j$-th column, rows $i$ to $j$, and $ij$-th entry,
respectively.  Denote by $B\in\R^{n \times m}$ the \emph{incidence}
matrix of a directed graph such that $B_{(i,j)}=1$ if the node $i$ is
the source of the edge $j$ and $B_{(i,j)}=-1$ if the node $i$ is the
sink node of the edge $j$; all other entries are zero.  The
\emph{Laplacian} matrix is $L=BYB^T$ where
$Y=\mathrm{diag}\{\{a_{ij}\}_{i,j\in\J}\}$, $a_{ij} \doteq y_{ij} E_i
E_j$, $y_{ij}$ denoting the pure imaginary $ij$-th line admittance and
$E_i$ denoting the bus voltage magnitude.  For connected graphs, $\ker
B^T = \ker L = \1_n$. The entries of the $m \times 1$ vector function
$\mathbf{f}=[f(\theta_i-\theta_j)]_{i,j\in\J}$ contain the scalar
function $f(\cdot)$ applied to $(\theta_i-\theta_j)$ in the same order
as the entries in the matrix $Y$. The symbol $\otimes$ denotes the
Kronecker product of matrices.  
$\R_{+0}^n$ denotes the set of real $n$-vectors with nonnegative
components. $\Z_+$ denotes the set of positive integers.
Inequalities and absolute values are taken componentwise.
A nonnegative vector function $T:\R_{+0}^n \to \R_{+0}^n$ is said to
be componentwise non-increasing (CNI) if whenever
$z_1,z_2\in\R_{+0}^n$ and $z_1 \le z_2$, then $T(z_1)\le T(z_2)$.

\section{Decentralised Droop Control Model}
We start by presenting our model of an inverter-based microgrid under
decentralised droop control, and then analyse its structure to reveal
important modal characteristics of the underlying linear part of the
system.  The model is essentially a weighted graph where each node
represents a common-voltage point of power injection, and branches
represent microgrid node-interconnecting lines
\cite{JWSP-FD-FB:12u,AiG13a}.

The standard primary droop control at each inverter $i$ in the
microgrid is such that the deviation in frequency $\dot{\theta}_i$ from
a nominal rated frequency $\omega^*$ is proportional to the power
injection $P_{e,i}$ in the following way:
\begin{equation}
	\label{eq:01}
	d_i \dot{\theta}_i = P_i^* - P_{e,i}
\end{equation}
where $d_i>0$ is the droop controller coefficient, 
$P_i^* \doteq P_{ref,i}-P_{L,i}$ is the inverter power injection error between
the inverter nominal injection setpoint $P_{ref,i}$ and the bus load
$P_{L,i}$, and $\omega_i = \omega^* + \dot\theta_i$ is the frequency of 
the voltage signal at the $i$-th inverter.  By assuming purely (loseless)
inductive lines, the power injection to each bus has the form
\begin{equation}
	\label{eq:02}
	P_{e,i}=\sum_{j=1}^{n} a_{ij} \sin(\theta_i-\theta_j),
\end{equation}
with $a_{ij} \doteq y_{ij} E_i E_j$, $y_{ij}$ denoting the pure
imaginary $ij$-th line admittance and $E_i$ denoting the bus voltage
magnitude. We make the standard \emph{decoupling approximation}
\cite{ZhH13} where all voltage magnitudes $E_i$ are constant so that
the power injection is considered a function of only the phase angles,
that is, $P_{e,i} =P_{e,i}(\theta)$.

The droop controller \eqref{eq:01} results in a static error in the steady
state frequency.  In \cite{AiG13a}, it is shown that as long as the
network state trajectories remain in a specified region, then the
controller in \eqref{eq:01} ensures network synchronisation to the
average frequency error
\begin{align}
	\label{eq:wsync1}
	\omega_{sync}=
		\frac{\sum_{i=1}^n d_i \dot{\theta}_i}{\sum_{i=1}^n d_i}
		= \frac{\sum_{i=1}^n P_i^* }{\sum_{i=1}^n d_i},
\end{align}
where the last equality follows from the fact that $\sum_{j=1}^{n}
P_{e,i} = 0$ for purely inductive lines.

We observe that $\omega_{sync}=0$ if and only if $\sum_{i=1}^n P_i^*=0$
or equivalently $\sum_{i=1}^n P_{ref,i} = \sum_{i=1}^n P_{L,i}$,
that is, the nominal injections are balanced. 
As discussed in~\cite{FD-JWSP-FB:14a}, it is not possible to achieve 
balanced nominal power injections since they depend on generally unknown and
variable load demand. Also, selecting the droop coefficients $d_i$
arbitrary large to make $\omega_{sync}$ small is not realistic.
Thus, complementary control action is required to eliminate  or 
at least reduce the frequency error $\omega_{sync}$; for example, 
by including additional secondary control inputs $p_i$ to each inverter
bus as follows:
\begin{align}
	\label{eq:03}
	d_i \dot{\theta}_i &= P_i^*-\sum_{j=1}^{n} a_{ij}\sin(\theta_i-\theta_j)-p_i,
	\\
	\label{eq:04}
	k_i \dot{p}_i &= \dot{\theta}_i - \epsilon p_i,
\end{align}
for each $i \in \I$ with $k_i,\epsilon>0$.
As shown in \cite{AiG13} and  discussed here in Section 
\ref{subsec:wsyncbndd}, the parameter $\epsilon$ in \eqref{eq:04}
can be tuned to reduce the frequency error. 

\begin{assu}
  In this paper we take all the droop coefficients as well as all the secondary
  control coefficients to be identical, that is, $d_i=d$ and $k_i=k$ for all 
  $i\in\I$.
\end{assu}
The above assumption leads to having a simplified expression for the
average frequency error which is
\begin{equation}
	\label{eq:wsync}
	\omega_{sync} 
	=\frac{\sum_{i=1}^n \dot{\theta}_i}{n}.
\end{equation}

Let $\sin(x)=x+f(x)$ where $f(x) \doteq \sin(x)-x$.
Then, from the definitions of the incidence matrix $B$
and the Laplacian matrix $L=BYB^T$ introduced in Notation and Definitions above, the
system \eqref{eq:03}--\eqref{eq:04} can be expressed as
\begin{equation}
	\label{eq:06}
        \dot{x}=Ax+H \mathbf{f}+\bar{P}
\end{equation}
where $x=[x_\theta^T, x_p^T]^T$, $x_\theta=[\theta_1,\dots,\theta_n]^T$, 
$x_p=[p_1,\dots,p_n]^T$, $\mathbf{f}=[f(\theta_i-\theta_j)]_{i,j\in\J}$ and the 
matrices
\begin{equation}
	\label{eq:07}
	A=\begin{bmatrix}
		\frac{-1}{d}  L  & \frac{-1}{d}  I_n \\[2mm]
		\frac{-1}{dk} L  & \frac{-e}{dk} I_n
	\end{bmatrix},
	H=\begin{bmatrix} 
		\frac{-1}{d} B Y  \\[2mm] \frac{-1}{dk} B Y
	  \end{bmatrix},
	\bar{P}=\begin{bmatrix} 
		\frac{1}{d} P^* \\[2mm] \frac{1}{dk} P^*
	\end{bmatrix}
\end{equation}
where $P^*=[P_1^*\; \dots \; P_n^*]^T$ and $e \doteq 1+\epsilon d$.

Let $(\mu_i,u_i)$, $i\in \mathcal{I}$ be the eigenvalue-eigenvector
pairs of the Laplacian matrix $L$ and define the associated eigenvalue
and eigenvector matrices as
\begin{align}
	\label{eq:08}
	M \doteq \diag \{\mu_1, \dots, \mu_n\}, \quad
	U \doteq [u_1\; \dots \; u_n].
\end{align}

The following properties of the Laplacian eigenstructure will
be useful for our later developments.

\begin{lemma}(\emph{Properties of the Laplacian eigenstructure})
  \label{lem:Wprop}
  The Laplacian eigenstructure~\eqref{eq:08} has the
  following properties:
  \begin{enumerate}[(a)]
  \item \label{item1} $\mu_1=0$ and $U_{(:,1)}=u_i=\1_n$ (due to the
    fact that $\ker B^T = \ker L = \1_n$ for connected graphs).
  	\item \label{item2} $[U^{-1}]_{(1,:)}=\1_n^T/n $ (since 
  			$[U^{-1}]_{(1,:)} U_{(:,1)}=[U^{-1}]_{(1,:)} \1_n = 1$).
  	\item \label{item3} $\sum_{i=1}^n U_{(i,j)}=0$, i.e. 
  			$\1_n^T U_{(:,j)}=0$ for $j=2,\dots,n$ (since 
  			$[U^{-1}]_{(1,:)} U_{(:,j)}=\1_n^T U_{(:,j)}/n = 0$).
    \item \label{item4} $\sum_{j=1}^n [U^{-1}]_{(i,j)}=0$, i.e.  
    		$[U^{-1}]_{(i,:)} \1_n =0$ for $i=2,\dots,n$ (since
            $[U^{-1}]_{(i,:)} U_{(:,1)}=[U^{-1}]_{(i,:)}\1_n =0$).
  \end{enumerate}
  \mer
\end{lemma}

The eigenstructure of the system \eqref{eq:06}--\eqref{eq:07} can be
conveniently represented in terms of the eigenstructure of the
Laplacian matrix, as shown in the following theorem.
\begin{theorem}
  \label{thm:A_eigen}
  For system \eqref{eq:06}--\eqref{eq:07}, the eigenvalues $\lambda_i$
  and eigenvectors $v_i$ of the matrix $A$ have the form
  
  \begin{align}
	\label{eq:11}
	\lambda_{2i-1,2i}&=
	  \begin{smallmatrix} 
	    -\frac{e+\mu_i k \mp R_i}{2dk}
	  \end{smallmatrix},
	R_i \doteq 
	  \begin{smallmatrix} 
	    \sqrt{4\mu_i k + (e-\mu_i k)^2},  i \in \I
	  \end{smallmatrix}
	\\
	\label{eq:12}
	[v_1\; v_2]&=\left[\begin{smallmatrix}
	   1 & \frac{k}{e} \\ 0 & 1 \end{smallmatrix}\right] \otimes u_1, \\
	\label{eq:13}
	[v_{2i-1}\;v_{2i}]
	&= \left[\begin{smallmatrix}
		\frac{e+dk\lambda_{2i-1}}{\mu_i} & \frac{e+dk\lambda_{2i}}{\mu_i} \\
		-1 & -1
	\end{smallmatrix}\right] \otimes u_i, 
	\; i \in \I-\{1\}
  \end{align}
  where $\mu_i$ and $u_i$'s are eigenvalues and eigenvectors of the Laplacian
  matrix $L$, respectively.
\end{theorem}

\begin{proof}
  An eigenvalue-eigenvector pair of the matrix $A$ satisfies 
  $(A-\lambda_i I_n)v_i=0$, that is, from \eqref{eq:07},
  \begin{align}
  	\label{eq:15}
  	\begin{bmatrix}
  		(-1/d)L-\lambda_i I_n  & (-1/d)I_n \\
  		(-1/dk)L & (-e/dk - \lambda_i)I_n
  	\end{bmatrix}
  	\begin{bmatrix} v_{i,\theta} \\ v_{i,p} \end{bmatrix}
  	= \0_{2n}
  	\end{align}
        where each eigenvector in \eqref{eq:12} and \eqref{eq:13} is
        partitioned into two $n \times 1$ vectors according to the
        structure of~$A$. Then, from the first $n$ rows of
        \eqref{eq:15}, $v_{i,p}$ can be written as
  \begin{equation}
  	\label{eq:16}
  	v_{i,p}=-(L+d\lambda_i I_n) v_{i,\theta},
  \end{equation}
  and hence, from the second group of $n$ rows in \eqref{eq:15} we
  obtain
  \begin{equation}
  	\label{eq:17}
  	\left(-L+(e+dk\lambda_i)(L+d\lambda_i I_n)\right) v_{i,\theta} = \0_n.
  \end{equation}
  We consider the eigenvalue-eigenvector pairs proposed in the
  statement of the theorem. The proof proceeds by first verifying that
  $(\lambda_i,v_{i,\theta})$ satisfy~\eqref{eq:17}.
  Then, the affirmed pair is replaced in \eqref{eq:16} to confirm the
  proposed expression for~$v_{i,p}$. 
  
  We first investigate the eigenstructure related to the first two
  eigenvalues where substituting $\mu_1=0$ [see
  Lemma~\ref{lem:Wprop}(\ref{item1})] into \eqref{eq:11} yields
  $\lambda_1=0$ and $\lambda_2=-e/(dk)$.
  
  Associated with $\lambda_1=0$ is the eigenvector
  $v_1=[v_{1,\theta}^T\;v_{1,p}^T]^T=[u_1^T\;\0_n^T]^T$. 
  Replacing $(\lambda_1,v_{1,\theta})$ in \eqref{eq:17} yields
  \begin{equation*}
  	(-L+Le) u_1 = (e-1) L u_1 = \0_n,
  \end{equation*}
  where the above is obtained on account of $L u_1 = \mu_1 u_1 =
  \0_n$.  Then, from \eqref{eq:16} we have $v_{1,p}=-L u_1 = \0_n$
  which confirms the validity of $(\lambda_1,v_1)$ as an
  eigenvalue-eigenvector pair of the matrix $A$.
  
  For the second eigenvalue of the matrix $A$, from \eqref{eq:12},
  corresponding to $\lambda_2=-e/(dk)$ we have $v_2=[v_{2,\theta}^T \;
  v_{2,p}^T]^T=[(k/e)u_1^T\;u_1^T]^T$.  Substituting
  $(\lambda_2,v_{2,\theta})$ into \eqref{eq:17} results in
  \begin{multline*}
  	[-L+\underbrace{(e+dk(-e/dk))}_0(L+d(-e/dk) I_n)] u_1 (k/e) \\
  	=-L u_1 (k/e) = \0_n
  \end{multline*}
  where we used $L u_1 = \mu_1 u_1 = 0$. From \eqref{eq:16}, $v_{2,p}$ is then
  \begin{align*}
  	v_{2,p}&=-(L+d(-e/dk) I_n) u_1 (k/e) \\
  	&= -L u_1 (k/e) + (e/k)u_1(k/e) = u_1
  \end{align*}
  which together with $v_{2,\theta}$ show the validity of~$(\lambda_2,v_2)$.
  
  Next, we show that for the remaining eigenvalues, the eigenvectors proposed
  in~\eqref{eq:13} satisfy \eqref{eq:16} and \eqref{eq:17}.  For
  simplicity we drop the subindex $i$ and write
  \begin{align}
  	\label{eq:18}
  	\lambda&=-\frac{e+\mu k \mp R}{2dk},\quad
  	R=\sqrt{4\mu k + (e-\mu k)^2} \\
  	\label{eq:19}
  	v&=\begin{bmatrix}
  		 u (e+dk\lambda)/\mu  \\ -u
  	\end{bmatrix}=
  	\begin{bmatrix}
  		 v_\theta \\ v_p
  	\end{bmatrix}.
  \end{align}
  Substituting the pair $(\lambda,v_\theta)$ into \eqref{eq:17} and
  disregarding the coefficient of $u$ in $v_\theta$ as it has no
  impact on the zero result lead to
  \begin{align}
  	\notag &(-L+(e+dk\lambda)(L+d\lambda I_n)) u \\
  	\notag &=-L u + (e + dk \lambda) (Lu + d \lambda u) \\
  	\notag &=-\mu u + (e + dk \lambda) (\mu u + d \lambda u) \\
  	\notag &=-\mu u + (e + dk \lambda) (\mu + d \lambda) u \\
  	\notag &=-\mu u + (\frac{e-\mu k \pm R}{2})
  			    (\frac{-e+\mu k \pm R}{2k}) u \\
  	\notag &=-\mu u + \frac{1}{4k} ( R^2 - (e-\mu k)^2) u \\
  	\notag &=-\mu u + \frac{1}{4k} ( 4 \mu k + (e-\mu k)^2) - (e-\mu k)^2) u \\
  	 &=-\mu u + \frac{1}{4k} ( 4 \mu k) u = \0_n.
  	\label{eq:20}
  \end{align}
  Then, \eqref{eq:16} is proven as follows:
  \begin{align}
    \notag v_p &= -(L+d \lambda I_n)v_\theta
    \\
    \notag &= -(Lu + d \lambda u) \frac{(e+dk\lambda)}{\mu} \\
    \notag &= -(\mu u + d \lambda u) \frac{(e+dk\lambda)}{\mu} \\
  	\label{eq:21} 
  	&= -\underbrace{(\mu + d \lambda) (e+dk\lambda)}_{\mu} \frac{u}{\mu}
  	= -u.
  \end{align}
  From \eqref{eq:20} and \eqref{eq:21}, it is clear that $(\lambda,v)$
  in \eqref{eq:18}, \eqref{eq:19} is an eigenvalue-eigenvector pair of
  the matrix $A$.
\end{proof}

Theorem~\ref{thm:A_eigen} derived expressions for the eigenvalues and
eigenvectors of the matrix $A$ in the microgrid
model~\eqref{eq:06}--\eqref{eq:07}. 
Through the obtained eigenstructure of the microgrid, one can exploit
a change into modal coordinates to investigate the system stability
properties. Define the associated matrices
\begin{align*}
	\Lambda &\doteq \diag \{ \lambda_1,\dots,\lambda_{2n} \}, \quad
	V \doteq \begin{bmatrix} v_1 & \dots & v_{2n} \end{bmatrix}.
\end{align*} 
We then consider the state transformation $x=Vz$.
From~\eqref{eq:06}--\eqref{eq:07} and noting that $\Lambda=V^{-1}AV$,
the transformed state $z$ satisfies
\begin{equation}
	\label{eq:22}
	\dot{z}=\Lambda z + V^{-1} H \mathbf{f} + V^{-1} \bar{P}
\end{equation}
where, by direct computation, 
\begin{equation}
	\label{eq:23}
	V^{-1}H = \Gamma U_H, \quad
	V^{-1}\bar{P} = -\Gamma U_P,
\end{equation}
with
\begin{align}
	\label{eq:24}
	\Gamma&=\mathrm{diag}\{
	  \begin{smallmatrix}
	     e-1,-\lambda_2,\lambda_3,-\lambda_4,\dots,
			\lambda_{2n-1},-\lambda_{2n}
	  \end{smallmatrix}\}, \\
	\label{eq:25}
	U_H&= u_h \otimes \left[\begin{smallmatrix} 1 \\ 1 \end{smallmatrix}\right]
	, \qquad 
	U_P= u_p \otimes \left[\begin{smallmatrix} 1 \\ 1 \end{smallmatrix}\right]
	\\
	\label{eq:27}
    u_h &= R^{-1} U^{-1} B Y
    , \quad
    u_p = R^{-1} U^{-1} P^*
\end{align}
where $u_h\in \R^{n \times m}$, $u_p\in \R^{n \times 1}$ and
$R=\mathrm{diag}\{R_i\}_{i\in\I}$.
We will show in the following section that the transformed
model~\eqref{eq:22}--\eqref{eq:27} has a special structure convenient for stability analysis.

\section{Stability Analysis}

The closed-loop system \eqref{eq:22} can be regarded as a linear
system with a nonlinear `perturbation' term (the second term) affected
by bounded disturbances (the third term).  Under certain conditions on
the nonlinear term one can expect the linear part of the dynamics to
dominate and, if the latter is stable, ultimately boundedness of the
trajectories starting inside a region of the state space may be
achieved~\cite{Kha02}.
In this regard, in this section we start by addressing the stability
of the linear part of system \eqref{eq:22} and follow progressive
steps to finally establish the ultimate boundedness of the
trajectories of the full nonlinear system, thus providing stability
conditions that go beyond local stability around the equilibrium
point.
It is worth noting that another analysis that considers a model
including nonlinearities in power systems has been presented 
in \cite{siljak2004}.

\subsection{Stability of the System's Linear Part}
\label{subsec:LinSt}
To begin with, the stability of the linear part of
system \eqref{eq:22} is established by analysing its  eigenvalues.
\begin{lemma}
  \label{lem:A_eval}
  The matrix $A$ in~\eqref{eq:06} (equivalently, $\Lambda$ in
  \eqref{eq:22}) has stable (real negative) eigenvalues, except for
  $\lambda_1=0$ which represents the rotational symmetry of the
  system.
\end{lemma}

\begin{proof}
  As can be seen in \eqref{eq:11}, the eigenvalues of the matrix $A$
  are functions of the eigenvalues of the Laplacian matrix $L$. It is
  well-known that the Laplacian matrix is a positive semi-definite
  matrix and hence, its eigenvalues $\mu_i$ are positive except for
  the zero eigenvalue $\mu_1=0$ representing the rotational symmetry.
  
  Each eigenvalue $\lambda_i$, $i \ne 1$ is stable if and only if
  \begin{align*}
  	-\frac{e+\mu_i k \mp \sqrt{4\mu_i k + (e-\mu_i k)^2}}{2dk} <
        0 \quad &\iff \\
  	\mp \sqrt{4\mu_i k + (e-\mu_i k)^2} < e+\mu_i k \quad &\iff \\
  	4\mu_i k + (e-\mu_i k)^2 < (e+\mu_i k)^2 \quad &\iff \\
  	4\mu_i k < 4\mu_i k e \quad &\iff
  	1 < e ,
  \end{align*}
  which is always true since $e=1+\epsilon d >1$ for $\epsilon,d~>~0$.
  Therefore, apart from the zero eigenvalue $\lambda_1=0$, prevalent
  to systems with the Laplacian matrix representation, the eigenvalues of the matrix $A$ are real
  negative numbers, thus stable.
\end{proof}

\subsection{Model Decoupling Property}
\label{subsec:decoup}
In view of facilitating the stability analysis, the structure of the
closed-loop system \eqref{eq:22}--\eqref{eq:27} can be unfolded one
step further by using a property of the eigenvector matrix of the
Laplacian $L$, as per the following remark.
\begin{remark}
  \label{rem:uhup}
  According to Lemma~\ref{lem:Wprop}\eqref{item2}, the first rows of the 
  matrices $u_h$ and $u_p$ in \eqref{eq:27}--\eqref{eq:27} are, respectively,
  \begin{itemize}
	\item	${u_h}_{(1,:)}=[U^{-1}]_{(1,:)} B Y /R_1 =\0_m^T$,
	\item	${u_p}_{(1)}=[U^{-1}]_{(1,:)} P^*/R_1 = (\sum_{i=1}^n P_i^*)/ne$,
  \end{itemize}
  where we have also used the structure of the incidence matrix~$B$
  and $R_1=e$ (see~\eqref{eq:11} for $\mu_1=0$).  \mer
\end{remark}

Letting $z=[z_1\ z_2\ \hat{z}^T]^T$, $\hat{z}=[z_3\ \dots \
z_{2n}]^T$, using~\eqref{eq:24} with $\lambda_2=-e/(dk)$, and
exploiting Remark~\ref{rem:uhup}, we have
\begin{align}
	\label{eq:28}
	\left[\begin{matrix} \dot{z}_1 \\ \dot{z}_2 \end{matrix}\right]
	&=\left[\begin{matrix} 
		0 & 0 \\ 0 & \lambda_2 
	 \end{matrix}\right]
	 \left[\begin{matrix} z_1 \\ z_2 \end{matrix}\right]
	+\left[\begin{matrix} (e-1)/(de) \\ 1/(dk) \end{matrix}\right]
	 \begin{matrix} \frac{(\sum_{i=1}^n P_i^*)}{n} \end{matrix}
	\\
	\label{eq:29}
	\dot{\hat{z}} &= \hat{\Lambda} \hat{z}+ \hat{\Gamma} (
        \hat{U}_H \mathbf{f} - \hat{U}_P )
\end{align}
where
\begin{align}
  \label{eq:Lhat}
  &\hat{\Lambda}=\diag(\lambda_3,\lambda_4, \dots,\lambda_{2n}), \\
  \label{eq:Ghat}
  &\hat{\Gamma}=\diag(\lambda_3,-\lambda_4,\dots,\lambda_{2n-1},-\lambda_{2n}), \\
  &\hat{U}_H=[U_H]_{(3:2n,:)}, \quad \hat{U}_P=[U_P]_{(3:2n)}. \label{eq:Uhat}
\end{align}

In the next step, the two subsystems \eqref{eq:28} and \eqref{eq:29}
are shown to be decoupled from each other. To this purpose, we study
the dependency of the function $\mathbf{f}$ on the $z$ states.

\begin{lemma}
  \label{lem:decoup_sys}
  The system \eqref{eq:29} consisting of the last $2n-2$ $z$ states is
  decoupled from the system \eqref{eq:28}.
\end{lemma}
\begin{proof}
  From $x=Vz$, if the matrix $V$ with columns given
  by~\eqref{eq:12}--\eqref{eq:13} is partitioned as
  \begin{align}
  	\label{eq:33}
	V=\left[\begin{smallmatrix} V_\theta \\ V_p \end{smallmatrix}\right]
	=\left[\begin{smallmatrix}
		u_1 & (k/e) u_1 & \dots \\
		0   & u_1 & \dots
	\end{smallmatrix}\right],
  \end{align}
  where $V_\theta,V_p \in \R^{n \times 2n}$, yields $x_\theta=V_\theta
  z$. Then we have, using the structure of the incidence matrix~$B$,
  \begin{equation}
	\label{eq:34}
	[\theta_i-\theta_j]_{i,j\in\J} = B^T V_\theta z.
  \end{equation}
  Using $u_1=\1_n$ (see Lemma~\ref{lem:Wprop}\eqref{item1}) and the
  fact that the matrix $B^T$ has just two nonzero elements $\{-1,1\}$
  in each of its rows, the first two columns of the matrix $B^T
  V_\theta$ are always zero and hence, \eqref{eq:34} does not depend
  on $(z_1,z_2)$.  That is,
  $\mathbf{f}=[f(\theta_i-\theta_j)]_{i,j\in\J}$ does not depend on
  $(z_1,z_2)$ and thus, system \eqref{eq:29} is decoupled from system
  \eqref{eq:28}.
\end{proof}

From Lemma~\ref{lem:A_eval} the linear subsystem \eqref{eq:28} has one
zero and one stable eigenvalue.  According to Lemma~\ref{lem:A_eval}
and Lemma~\ref{lem:decoup_sys}, the companion subsystem \eqref{eq:29}
has a stable diagonal linear part and a nonlinear perturbation term
that depends only on its own state variables. In the following two
sections we study the boundedness properties of these decoupled
subsystems.

\subsection{Boundedness of the Average Frequency Error}
\label{subsec:wsyncbndd}
The representation \eqref{eq:28}--\eqref{eq:29} of the microgrid
system facilitates the analysis of the average frequency error
and its boundedness, as shown next.
\begin{lemma}
  \label{lem:wsync_bndd}
  For the microgrid system represented by
  \eqref{eq:28}--\eqref{eq:29}, the average frequency error
  $\omega_{sync}$ given in \eqref{eq:wsync} is bounded if all the
  inverter power injection errors $P_i^*, i\in\I$, are bounded.
\end{lemma}
\begin{proof}
  From $x=V z$ it can be shown that 
  \begin{align}
  	\label{eq:32}
  	z_1 = (\sum_{i=1}^n \theta_i)/n - k z_2/e. \quad
  \end{align}
  Since from \eqref{eq:28}, $\dot{z}_2=\lambda_2 z_2+(\sum_{i=1}^n P_i^*)/(ndk)$
  with $\lambda_2=-e/(dk)$, then $z_2$ and $\dot{z}_2$ remain bounded for 
  bounded $P_i^*,i\in~\I$. Furthermore, from \eqref{eq:28} and \eqref{eq:32},
  we have 
  \begin{equation}
  	\label{eq:321}
    \dot{z}_1=\underbrace{(\sum_{i=1}^n \dot{\theta}_i)/n}_{\omega_{sync}} 
    - k \dot{z}_2/e  =(e-1)(\sum_{i=1}^n P_i^*)/(nde).
  \end{equation}
  Then the average frequency error \eqref{eq:wsync} also
  remains bounded for bounded inverter errors $P_i^*,i\in\I$.
\end{proof}

\begin{corollary}
  \label{cor:wsync}
  The average frequency error $\omega_{sync}$ converges to
  \begin{equation}
  	\label{eq:wsync_ss}
  	{\omega_{sync}}_{ss}=\frac{(\sum_{i=1}^n P_i^*)d \epsilon}{nd(1+d\epsilon)}
  \end{equation}
  if $\sum_{i=1}^n P_i^*$ is constant.  
\end{corollary}
\begin{proof}
  From the $z_2$ equation in \eqref{eq:28}, $z_2$ is proved to converge to 
  a constant and hence $\dot{z}_2$ converges to zero. Then, from
  \eqref{eq:321} and $e=1+\epsilon d$, $\omega_{sync}$ converges to 
  \eqref{eq:wsync_ss}.
\end{proof}

\begin{remark}
  From Corollary~\ref{cor:wsync}, a smaller value of $\epsilon$
  yields a smaller average frequency steady state error.
  \mer
\end{remark}

\subsection{Ultimate Boundedness}
\label{subsec:UB}
We will analyse the ultimate boundedness properties of the subsystem
\eqref{eq:29} by applying Theorem~3 of~\cite{HaS13}. When specialised
to non-switched systems, the latter result establishes that for a
stable linear system with a nonlinear perturbation term, the trajectories
starting inside a region of the state space are ultimately bounded if
the nonlinear perturbation satisfies certain conditions.  
More specifically, to meet the requirements of \cite[Theorem~3]{HaS13},
the perturbation term should be bounded by a componentwise non-increasing
(CNI) function and further satisfy a contractivity condition. 
We first derive in the following result a CNI
bound for the perturbation term in~\eqref{eq:29} and then address the
contractivity condition in Lemma~\ref{lem:contractivity}.

\begin{lemma}
  \label{lem:CNIfcn}
  The perturbation term $\hat{\Gamma} (\hat{U}_H
  \mathbf{f}-\hat{U}_P)$ in system~\eqref{eq:29} is bounded by a CNI
  function as follows:
  \begin{equation}
    \label{eq:1}
    |\hat{\Gamma}(\hat{U}_H \mathbf{f}-\hat{U}_P)|\le |\hat{\Gamma}\hat{U}_H|
    F(\hat{z}) +  |\hat{\Gamma} \hat{U}_P|, 
  \end{equation}
where
\begin{equation}
  \label{eq:2}
  F(\hat{z}) \doteq  \frac{(|B^T V_\theta||z|)^3}{6},
\end{equation}
  with $V_\theta$ as in \eqref{eq:33}.
\end{lemma}
 \begin{proof}
  We first  bound the nonlinear function $\mathbf{f}$, with
  components $f(\theta_i-\theta_j)$ with $f(x)=\sin(x)-x$.
  Recalling from the proof of Lemma~\ref{lem:decoup_sys}
  that~\eqref{eq:34} only depends on $\hat{z}$, and using the
  inequality $|\sin(x)-x|\le|x|^3/6$ we can bound
  \begin{equation}
	\label{eq:36}
	f(\theta_i-\theta_j) 
	\le \frac{(|{[B^T]}_{(i,:)} V_\theta||z|)^3}{6} 
	 \doteq F_i(\hat{z}),
  \end{equation}
  yielding
  \begin{equation}
	\label{eq:37}
	\mathbf{f}=[f(\theta_i-\theta_j)]_{i,j\in\J}
	\le \frac{(|B^T V_\theta||z|)^3}{6} \doteq F(\hat{z}).
  \end{equation}
  The bound \eqref{eq:1}--\eqref{eq:2} then follows. The CNI property
  of the bound is immediate from the nonnegativity of all entries in
  the products involved.
\end{proof}

Following \cite{HaS13}, we next define a nonlinear mapping
$T:\R_{+0}^{2n-2}\rightarrow \R_{+0}^{2n-2}$ constructed from the
bound~\eqref{eq:1} as follows:
 \begin{align}
   \notag T(\hat{z}) &\doteq |\hat{\Lambda}|^{-1} (|\hat{\Gamma}
   \hat{U}_H| F(\hat{z}) + |\hat{\Gamma} \hat{U}_P|)
   \\
   \label{eq:38}
   &= |\hat{U}_H| F(\hat{z}) + |\hat{U}_P|,
  \end{align}
where the second line follows from \eqref{eq:Lhat} and \eqref{eq:Ghat}.
From \cite[Theorem~3]{HaS13} (see \cite{haimovich12:_bound-arxiv} for
proofs), if a vector $\bar{z}$ with positive components exists such
that contractivity condition
\begin{equation}
	\label{eq:35}
	T(\bar{z}) < \bar{z}
\end{equation}
holds componentwise, then the trajectories of the nonlinear system
\eqref{eq:29} are ultimately bounded and the ultimate bound can be
found by recursively iterating the mapping $T(\cdot)$ starting from
$\bar{z}$. In the following lemma we give a sufficient condition
for~\eqref{eq:35} to hold for some~$\bar{z}$.
 
\begin{lemma}
  \label{lem:contractivity}
  Suppose there exist positive constants $g_1$, $g_2,\dots,g_{n-1}$ such
  that the scalar inequality
  \begin{equation}
  	\label{eq:ContCond}
  	{u_p}_{(i+1)}^2 < \frac{4 g_i^3}{27 \gamma_i}
  \end{equation}
  holds for $i=1,\dots,n-1$, where $u_p$ is defined in \eqref{eq:27},
  $\gamma_i \doteq |{u_h}_{(i+1,:)}|(|B^T V_\theta|G)^3/6 >0$, with
  $u_h$ defined in \eqref{eq:27} and $G \doteq
  [1,1,g_1,g_1,g_2,g_2,\dots,g_{n-1}]^T$. Then there exists a scalar
  $\zeta>0$ such that the nonnegative vector $\bar{z}\doteq G_{(3:2n)}
  \zeta $ satisfies the contractivity condition~\eqref{eq:35}.
\end{lemma}

\begin{proof}
From \eqref{eq:25}, it can be seen that each even
row of $U_H$ and $U_P$ is equal to its preceding row and thus, $\hat
U_H$, $\hat U_P$ defined in~\eqref{eq:Uhat} and the vector function
$T(\hat{z})$ in~\eqref{eq:38} also share the same property. That is,
letting $T(\hat{z})=
\left[\begin{smallmatrix}
  T_1(\hat{z}) & T_2 (\hat{z}) & \dots & T_{2n-2}(\hat{z}) 
\end{smallmatrix}\right] $, 
we have for $i=1,\dots,n-1$
  \begin{align}
  	\label{eq:39}
    \left[\begin{smallmatrix}
      	T_{2i-1}(\hat{z}) \\ T_{2i}(\hat{z}) 
    \end{smallmatrix}\right]
    &= t_i(\hat{z}) \otimes 
    \left[\begin{smallmatrix} 1 \\ 1 \end{smallmatrix}\right], 
    \\
    \notag
    t_i(\hat{z}) &= |{u_h}_{(i+1,:)}| F(\hat{z}) + |{u_p}_{(i+1)}|
    \\
    \label{eq:391}
    &= |{u_h}_{(i+1,:)}| \frac{(|B^T V_\theta||z|)^3}{6} +
    |{u_p}_{(i+1)}|.
  \end{align}
  The contractivity condition \eqref{eq:35} with the consideration of 
  \eqref{eq:39} takes the form
  \begin{align*}
  	\left[\begin{smallmatrix}
		T_{2i-1}(\bar{z}) \\ T_{2i}(\bar{z}) 
	\end{smallmatrix}\right]
	&= t_i(\bar{z}) \otimes 
	\left[\begin{smallmatrix} 1 \\ 1 \end{smallmatrix}\right]
	<
	\left[\begin{smallmatrix} \
		\bar{z}_{2i-1} \\ \bar{z}_{2i} 
	\end{smallmatrix}\right],
  \end{align*}
  which, by choosing $\bar{z}$ to have pairwise repeated rows, can be
  further simplified to
  \begin{align*}
  	\left[\begin{smallmatrix}
		T_{2i-1}(\bar{z}) \\ T_{2i}(\bar{z}) 
	\end{smallmatrix}\right]
	&= t_i(\bar{z}) \otimes 
	\left[\begin{smallmatrix} 1 \\ 1 \end{smallmatrix}\right]
	<
	\bar{z}_{2i} \otimes 
	\left[\begin{smallmatrix} 1 \\ 1 \end{smallmatrix}\right],
  \end{align*}
  and hence,
  \begin{equation}
	\label{eq:40}
	t_i(\bar{z}) 
	= |{u_h}_{(i+1,:)}| \frac{(|B^T V_\theta||z|)^3}{6} + |{u_p}_{(i+1)}| 
	< \bar{z}_{2i}
  \end{equation}
  for $i=1,\dots,n-1$.
  Further substituting $z= G\zeta$ and $\bar{z} =  G_{(3:2n)} \zeta
  $ with $G = [1,1,g_1,g_1,g_2,g_2,\dots,g_{n-1}]^T$, yields
  \begin{align}
	\label{eq:41}
        \bar{t}_i(\zeta) &=|{u_h}_{(i+1,:)}| \frac{(|B^T V_\theta|G)^3}{6}
        \zeta^3 + |{u_p}_{(i+1)}| < g_i \zeta
  \end{align}
  for $i=1,\dots,n-1$, where $\bar{t}_i(\zeta)
  =t_i(\bar{z})$. Equivalently,
  \begin{equation}
	\label{eq:43}
        \gamma_i \zeta^3 - g_i \zeta + |{u_p}_{(i+1)}| < 0
  \end{equation}
  where $\gamma_i=|{u_h}_{(i+1,:)}|(|B^T V_\theta|G)^3/6 >0$. 

  For a generic cubic function $Q(\zeta) \doteq
  a\zeta^3+b\zeta^2+c\zeta+d$, it is known that to have $Q(\zeta)<0$
  for $\zeta>0$, $Q(\zeta)$ must have three distinct real roots, which
  is guaranteed if its discriminant
  $\Delta=18abcd-4db^3+b^2c^2-4ac^3-27a^2d^2$ is positive.  For the
  cubic function on the left hand side of~\eqref{eq:43} the positive
  discriminant condition takes the form
  \begin{equation*}
    \Delta_i = \gamma_i (4  g_i^3-27 \gamma_i {u_p}_{(i+1)}^2) > 0,
  \end{equation*}
  which coincides with~\eqref{eq:ContCond}. 
\end{proof}

We observe that the contractivity condition~\eqref{eq:ContCond} can be
loosely interpreted as a tolerance on `how dissimilar' 
the inverter power errors, $P_i^*$,
are allowed to be to meet the desired
requirements. Indeed, from Lemma~\ref{lem:Wprop}\eqref{item4} and the
definition of $u_p$ in \eqref{eq:27}, for $i=1:n-1$, 
\begin{equation}
  \label{eq:3}
  u_p(i+1) =\ell_{i}(P_1^*, \dots, P_n^*)
\end{equation}
is a linear combination of the inverter power errors
such that, if
$P_i^*=P_j^*$ for all $i,j\in\I$ we have $u_p(i+1)=0$ for $i=1:n-1$
and condition~\eqref{eq:ContCond} is automatically satisfied.


We now have all the elements to establish the stability properties of
the droop controlled microgrid system.

\begin{theorem}
  Under the conditions of Lemma~\ref{lem:contractivity}, let $\zeta>0$
  satisfy~\eqref{eq:43}. Then, for the microgrid system represented by
  \eqref{eq:28}--\eqref{eq:29}, the average frequency error
  $\omega_{sync}$ given in \eqref{eq:wsync} is bounded and the
  trajectories of subsystem~\eqref{eq:29} with initial conditions
  satisfying $|\hat{z}(0)| \le G\zeta$ are ultimately bounded as $\lim
  \, \sup_{t\to \infty} |\hat{z}(t)| \le \lim_{k \to \infty}
  T^k(G\zeta)$.
\end{theorem}
\begin{proof}
  Immediate from the results in this section and Theorem~3 of~\cite{HaS13}.
\end{proof}

\section{Example}
To illustrate the discussed concepts, we consider an academic example of
a microgrid system consisting of three inverter buses and two edges with $a_{12}=2$,
$a_{13}=5$, $a_{23}=0$. For this system, the graph data, the
incidence matrix, the Laplacian matrix and its eigenstructure, after
removing all zero rows and columns corresponding to $a_{23}=0$, are
given by
\begin{align*}
	B &= \left[\begin{smallmatrix}
		 1 & 1 \\ -1 & 0 \\ 0 & -1
	\end{smallmatrix}\right],  \quad
	Y = \left[\begin{smallmatrix}
		 2 & 0 \\ 0 & 5
	\end{smallmatrix}\right],\quad
	L = \left[\begin{smallmatrix}
		 7 & -2 & -5 \\
    	-2 &  2 &  0 \\
    	-5 &  0 &  5
	\end{smallmatrix}\right], \\[2mm]
	M &= \diag\{\begin{smallmatrix}
		0, & 2.6411, & 11.3589 
		\end{smallmatrix}\}, \quad
	U =\left[\begin{smallmatrix}
		1 &   0.4718  & -1.2718 \\
    	1 &  -1.4718  &  0.2718 \\
    	1 &   1       &  1
	\end{smallmatrix}\right].
\end{align*}
Using the above data, the system matrices and its eigenvalue-eigenvectors
from \eqref{eq:07}, \eqref{eq:11}--\eqref{eq:13}, with $d=1$, $k=1$ 
and $\epsilon=1$ giving $e=1+\epsilon d=2$, are
\begin{align*}
	A  &= \left[\begin{smallmatrix}
	-7  &   2  &   5  &  -1  &   0  &   0\\
     2  &  -2  &   0  &   0  &  -1  &   0\\
     5  &   0  &  -5  &   0  &   0  &  -1\\
    -7  &   2  &   5  &  -2  &   0  &   0\\
     2  &  -2  &   0  &   0  &  -2  &   0\\
     5  &   0  &  -5  &   0  &   0  &  -2
	\end{smallmatrix}\right], \quad
	H  = \left[\begin{smallmatrix}
	-2  &  -5 \\  2  &   0 \\  0  &   5 \\
    -2  &  -5 \\  2  &   0 \\  0  &   5
	\end{smallmatrix}\right], \\[2mm]
	\Lambda &=\diag \{\begin{smallmatrix}
		0,& -2,& -0.6641,& -3.9770,& -0.9126,& -12.4463 
		\end{smallmatrix}\}, \\[2mm]
	V &= \left[\begin{smallmatrix}
		1 &  0.5 &   0.2386 &  -0.3532 & -0.1217  &  1.1696 \\
    	1 &  0.5 &  -0.7444 &   1.1017 &  0.0260  & -0.2499 \\
    	1 &  0.5 &   0.5058 &  -0.7486 &  0.0957  & -0.9197 \\
        0 &  1   &  -0.4718 &  -0.4718 &  1.2718  &  1.2718 \\
        0 &  1   &   1.4718 &   1.4718 & -0.2718  & -0.2718 \\
        0 &  1   &  -1      &  -1      & -1       & -1
	\end{smallmatrix}\right]
\end{align*}

Next, to form the transformed system \eqref{eq:22} with matrices
\eqref{eq:23}, the required matrices \eqref{eq:24}--\eqref{eq:28} are
\begin{align*}
	\Gamma &= \diag
		\{\begin{smallmatrix}
			1 , 2 , -0.6641 , 3.9770 , -0.9126 , 12.4463 
		\end{smallmatrix}\}, \\
	R &=\diag \{\begin{smallmatrix}
		2 , 3.3129 , 11.5336
		\end{smallmatrix}\}, \\
	u_h &=R^{-1} U^{-1} B Y \\
	&=R^{-1}
	\left[\begin{smallmatrix}
		1 &   0.4718  & -1.2718 \\
    	1 &  -1.4718  &  0.2718 \\
    	1 &   1       &  1
	\end{smallmatrix}\right]^{-1}
	\left[\begin{smallmatrix} 
		1 & 1 \\ -1 & 0 \\ 0 & -1
	\end{smallmatrix}\right]
	Y
	\\
	&=R^{-1}
	\left[\begin{smallmatrix}
		0 & 0 \\ 0.5735 & -0.1559 \\ -0.5735 & -0.8441
	\end{smallmatrix}\right]
	Y
	=\left[\begin{smallmatrix}
		 0      &  0      \\
    	 0.3462 & -0.2353 \\
   		-0.0995 & -0.3659
	\end{smallmatrix}\right],
	\\	u_p
	&=R^{-1} U^{-1} P^* 
	\\
	&=R^{-1}
	\left[\begin{smallmatrix}
		1 &   0.4718  & -1.2718 \\
    	1 &  -1.4718  &  0.2718 \\
    	1 &   1       &  1
	\end{smallmatrix}\right]^{-1}
	\left[\begin{smallmatrix}
		P_1^* \\ P_2^* \\ P_3^*
	\end{smallmatrix}\right] 
	\\
	&=R^{-1}
	\left[\begin{smallmatrix}
		 0.3333 &  0.3333 & 0.3333 \\
    	 0.1392 & -0.4343 & 0.2951 \\
   		-0.4726 &  0.1010 & 0.3716
	\end{smallmatrix}\right]
	\left[\begin{smallmatrix}
		P_1^* \\ P_2^* \\ P_3^*
	\end{smallmatrix}\right] 
	\\
	&=\left[\begin{smallmatrix}
		 0.1667 &   0.1667  &  0.1667 \\
    	 0.0420 &  -0.1311  &  0.0891 \\
   		-0.0410 &   0.0088  &  0.0322
	\end{smallmatrix}\right]
	\left[\begin{smallmatrix}
		P_1^* \\ P_2^* \\ P_3^*
	\end{smallmatrix}\right]
	\\
	&
	=\left[\begin{smallmatrix}
		(P_1^*+P_2^*+P_3^*)/6 \\ 
		 0.0420 P_1^* - 0.1311 P_2^* + 0.0891 P_3^* \\
		-0.0410 P_1^* + 0.0088 P_2^* + 0.0322 P_3^*
	\end{smallmatrix}\right],
	\\
	U_H&= u_h \otimes \left[\begin{smallmatrix} 1 \\ 1 \end{smallmatrix}\right],
	\\
	U_P&= u_p \otimes \left[\begin{smallmatrix} 1 \\ 1 \end{smallmatrix}\right],
	\\
	V^{-1} H  &= \Gamma U_H
	=\left[\begin{smallmatrix}
		0 & 0 \\ 0 & 0 \\
  	   -0.2299  &  0.1562 \\  1.3770  & -0.9356 \\
    	0.0908  &  0.3340 \\ -1.2378  & -4.5546
	\end{smallmatrix}\right],
	\\
	V^{-1} \bar{P}  &= - \Gamma U_P
	=\left[\begin{smallmatrix}
		(P_1^*+P_2^*+P_3^*)/6 \\ (P_1^*+P_2^*+P_3^*)/3 \\
		-0.0279 P_1^* + 0.0871 P_2^* - 0.0592 P_3^* \\
		 0.1670 P_1^* - 0.5214 P_2^* + 0.3544 P_3^* \\
		 0.0374 P_1^* - 0.0080 P_2^* - 0.0294 P_3^* \\
		-0.5103 P_1^* + 0.1095 P_2^* + 0.4008 P_3^*
	\end{smallmatrix}\right].
\end{align*}
It can be seen that ${u_h}_{(1,:)}=\0_2^T$ yields the first two rows of
$V^{-1}H$ equal to zero, which confirms that the two subsystems
\begin{align*}
	\left[\begin{smallmatrix} \dot{z}_1 \\ \dot{z}_2 \end{smallmatrix}\right]
	&=\left[\begin{smallmatrix} 0 & 0 \\ 0 & -2  \end{smallmatrix}\right]
	 \left[\begin{smallmatrix} z_1 \\ z_2 \end{smallmatrix}\right]
	+\left[\begin{smallmatrix} 1/2 \\ 1 \end{smallmatrix}\right]
	 \left[\begin{smallmatrix} (\sum_{i=1}^n P_i^*)/3 \end{smallmatrix}\right]
	\\
	\left[\begin{smallmatrix} 
		\dot{z}_1 \\ \dot{z}_2 \\ \dot{z}_3 \\ \dot{z}_3
	\end{smallmatrix}\right]
	&=\left[\begin{smallmatrix}
		 -0.6641 & 0 & 0 & 0 \\ 0 & -3.9770 & 0 & 0 \\
		 0 & 0 & -0.9126 & 0 \\ 0 & 0 & 0 & -12.4463
	\end{smallmatrix}\right]
	\left[\begin{smallmatrix} z_1 \\ z_2 \\ z_3 \\ z_4 \end{smallmatrix}\right] \\
	&+\left[\begin{smallmatrix} 1/2 \\ 1 \end{smallmatrix}\right]
	 \left[\begin{smallmatrix} (\sum_{i=1}^n P_i^*)/3 \end{smallmatrix}\right]
\end{align*}
 are decoupled.

Partitioning $V$ as in \eqref{eq:33} with $V_\theta=V_{(1:3,:)}$, 
the line phases are computed from \eqref{eq:34} to be
\begin{align*}
	\left[\begin{smallmatrix}
		\theta_1-\theta_2 \\ \theta_1-\theta_3
	\end{smallmatrix}\right]
	&=\tiny{ B^T V_\theta z }
	\\
	&=\left[\begin{smallmatrix} 
		0 & 0 & 0.9831 &-1.4549 &-0.1478 & 1.4195 \\
        0 & 0 &-0.2672 & 0.3954 &-0.2175 & 2.0893
    \end{smallmatrix}\right]
	\left[\begin{smallmatrix} 
		z_1 \\ z_2 \\ z_3 \\ z_4 \\ z_5 \\ z_6
	\end{smallmatrix}\right]
	\\
	&
	=\left[\begin{smallmatrix} 
		0.9831 &-1.4549 &-0.1478 & 1.4195 \\
        -0.2672 & 0.3954 &-0.2175 & 2.0893
    \end{smallmatrix}\right]
	\left[\begin{smallmatrix} 
		z_3 \\ z_4 \\ z_5 \\ z_6
	\end{smallmatrix}\right].
\end{align*}
The function $F(\hat{z})$ in Lemma~\ref{lem:CNIfcn} is obtained from
\eqref{eq:2} as
\begin{align*}
	  F(\hat{z})&=\frac{(|B^T V_\theta||z|)^3}{6} 
	\\
	&=\left(
	\left|\begin{smallmatrix} 
		0.9831 &-1.4549 &-0.1478 & 1.4195 \\
        -0.2672 & 0.3954 &-0.2175 & 2.0893
    \end{smallmatrix}\right|
    \left|\begin{smallmatrix} 
		z_3 \\ z_4 \\ z_5 \\ z_6
	\end{smallmatrix}\right|
	\right)^3/6 \\
    &=\left[\begin{smallmatrix} 
		(0.9831 z_3 + 1.4549 z_4 + 0.1478 z_5 + 1.4195 z_6)^3/6 \\
        (0.2672 z_3 + 0.3954 z_4 + 0.2175 z_5 + 2.0893 z_6)^3/6
    \end{smallmatrix}\right].
\end{align*}
Then, from \eqref{eq:38}, \eqref{eq:39}--\eqref{eq:391} the nonlinear
mapping $T:\R_{+0}^{4}\rightarrow \R_{+0}^{4}$ is
\begin{align*}
	T(\hat{z}) = |\hat{U}_H| F(\hat{z}) + |\hat{U}_P|
	=\left[\begin{smallmatrix}  t_1(\hat{z}) \\ t_2(\hat{z}) \end{smallmatrix}\right]
	\otimes \left[\begin{smallmatrix}  1 \\ 1 \end{smallmatrix}\right]
\end{align*}
with
\begin{multline*}
\hspace{-3mm}	t_1(\hat{z})= |{u_h}_{(2,:)}| F(\hat{z}) + |{u_p}_{(2)}| \\
	\hspace{6mm} =\left|\begin{smallmatrix} 
		0.3462 & -0.2353 
	\end{smallmatrix}\right|
	F(\hat{z})
    + \left|\begin{smallmatrix}
		 0.0420 P_1^* - 0.1311 P_2^* + 0.0891 P_3^*
	\end{smallmatrix}\right|, 
\end{multline*}	
\begin{multline*}
\hspace{-3mm}	t_2(\hat{z})= |{u_h}_{(3,:)}| F(\hat{z}) + |{u_p}_{(3)}| \\
	\hspace{3mm} 	=\left|\begin{smallmatrix}
   		-0.0995 & -0.3659
	\end{smallmatrix}\right|
	F(\hat{z})
    + \left|\begin{smallmatrix}
		-0.0410 P_1^* + 0.0088 P_2^* + 0.0322 P_3^*
	\end{smallmatrix}\right|.
\end{multline*}
With the selection of $z_3=z_4=g_1 \zeta$ and $z_5=z_6=g_2 \zeta$, the
scalar inequalities \eqref{eq:41} to satisfy the contractivity
condition~\eqref{eq:35} are
\begin{align*}
	t_1(\zeta)
	&=( \begin{smallmatrix} 
			0.3462(0.9831 g_1 + 1.4549 g_1 + 0.1478 g_2 + 1.4195 g_2)^3 
	   \end{smallmatrix} \\
	 &+\begin{smallmatrix}
	 		0.2353(0.2672 g_1 + 0.3954 g_1 + 0.2175 g_2 + 2.0893 g_2)^3 
	   \end{smallmatrix} 
	  ) \zeta^3 \\
	 &+\left|\begin{smallmatrix}
		 0.0420 P_1^* - 0.1311 P_2^* + 0.0891 P_3^*
	\end{smallmatrix}\right| \\
	&= \gamma_1 \zeta^3 + |{u_p}_{(2)}| < g_1 \zeta
	\\
	t_2(\zeta)
	&=( 
	   \begin{smallmatrix} 
			0.0995(0.9831 g_1 + 1.4549 g_1 + 0.1478 g_2 + 1.4195 g_2)^3 
	   \end{smallmatrix} \\
	 &+\begin{smallmatrix}
	 		0.3659(0.2672 g_1 + 0.3954 g_1 + 0.2175 g_2 + 2.0893 g_2)^3 
	   \end{smallmatrix} 
	  ) \zeta^3 \\
	 &+\left|\begin{smallmatrix}
		-0.0410 P_1^* + 0.0088 P_2^* + 0.0322 P_3^*
	\end{smallmatrix}\right| \\
	&= \gamma_2 \zeta^3 + |{u_p}_{(3)}| < g_2 \zeta
\end{align*}
for arbitrary $g_1,g_2>0$.
The inverter power injection setpoints, through the linear
functions~\eqref{eq:3}, then need to satisfy the scalar inequalities
\begin{align}
	\label{ex:01}
  {u_p}_{(2)}^2 & =[\ell_1(P_1^*,P_2^*,P_3^*)]^2< \frac{4 g_1^3}{27 \gamma_1}
  	\doteq b_1,\\
  	\label{ex:02}
  {u_p}_{(3)}^2 &= [\ell_2(P_1^*,P_2^*,P_3^*)]^2 < \frac{4 g_2^3}{27\gamma_2}
    \doteq b_2
\end{align}
for the system to be ultimately bounded.  With regard to these
inequalities, one can run a nonlinear optimisation on $g_1$ and $g_2$
to maximise the upper bounds $b_1$ and $b_2$.  The nonlinear
optimisation
\begin{equation*}
	\max \min_{g_1,g_2} \{b_1,b_2\}
\end{equation*}
yields $g_1=8.0377, g_2=6.4202$ which in turn lead to
$ b_1=0.0421, b_2=0.0421$.

Take, for instance, $P_1^*=1,P_1^*=2,P_1^*=3$. The contractivity conditions
\eqref{ex:01}--\eqref{ex:02} are then satisfied
\begin{align*}
  {u_p}_{(2)}^2 = 0.0471^2 &< 0.0421=b_1,\\
  {u_p}_{(3)}^2 = 0.0732^2 &< 0.0421=b_2.
\end{align*}

The next step is to find $\bar{z}=G \zeta$. For each $t_i(\bar{z})$
function, the $\zeta$ domain for which $t_i(\zeta)<g_i \zeta$ is the
interval between the two positive roots of the polynomial $Q_i(\zeta)
\doteq t_i(\zeta)-g_i \zeta=0$.  For $i=1,2$ we have
$\mathrm{roots}(Q_1)=\{-0.0665,0.0003,0.0662\}$ and
$\mathrm{roots}(Q_2)=\{-0.0835,0.0008,0.0827\}$ which yields
\begin{align*}
	\zeta(Q_1)=(0.0003,0.0662), \quad
	\zeta(Q_2)=(0.0008,0.0827).
\end{align*}
Then, the $\zeta$ domain that satisfies both conditions is the
intersection of these intervals, that is,
\begin{equation}
	\label{ex:03}
	\zeta \in \zeta(Q_1) \bigcap \zeta(Q_2) = (0.0008,0.0662).
\end{equation}
Now we just need to select a starting point $\zeta_0$ from this
interval, compute the associated $\bar{z}_0$ and iteratively calculate
the ultimate bound of the system.  From \cite{HaS13}, the ultimate
bound can be computed by first taking $\bar{z}_0=G_{(3:2n)}\zeta_0$,
$T^1(\bar{z})=T(\bar{z}_0)$ and then iterating,
$T^{k+1}(\bar{z})=T(T^k(\bar{z}))$ for $k\in\Z_+$. Since
$T^{k+1}(\bar{z})\le T^k(\bar{z})$, the ultimate bound is obtained as
$\lim_{k \to \infty} T^k (\bar{z})=\bnd_z  >0$.

Let $\zeta_0=0.0327$.  Then
$\bar{z}_0=[0.2628,0.2628,0.2099,0.2099]^T$.  The resulting ultimate
bound on the $\hat{z}$ states is
\begin{equation}
	\label{ex:04}
	\bnd_z = \left[\begin{smallmatrix}
		0.0481 \\  0.0481 \\  0.0739 \\  0.0739
	\end{smallmatrix}\right].
\end{equation}
We can interpret this ultimate bound 
on the line phases $[\theta_i-\theta_j]_{i,j\in\J}$ 
as follows
\begin{align*}
	\left|\begin{smallmatrix}
		\theta_1-\theta_2 \\ \theta_1-\theta_3
	\end{smallmatrix}\right|
	=|B^T V_\theta z|
	\le 
	|B^T V_\theta| |z|
	\le 
	|B^T V_\theta| \left[\begin{smallmatrix}
		* \\ * \\ \bnd_z
	\end{smallmatrix}\right]=
	\left[\begin{smallmatrix}
		 0.2331 \\ 0.2023
	\end{smallmatrix}\right].
\end{align*}
where the $*$ entries are irrelevant since the first two columns
of $B^T V_\theta$ are zero.
The above bounds on the phase differences is validated as can be seen in
Fig.~\ref{fig:subplot}(a).

The next variable derived from this simulation is the average frequency 
error as in \eqref{eq:wsync1}. Corollary~\ref{cor:wsync} proves that
this frequency converges to the steady state frequency ${\omega_{sync}}_{ss}$
as in \eqref{eq:wsync_ss}. It is also notable that this steady state average
frequency static error is reliant on $\epsilon$, that is, by decreasing 
$\epsilon$ we obtain a smaller ${\omega_{sync}}_{ss}$.
For $\epsilon=1$ and $\epsilon=0.1$ the obtained values are
\begin{align*}
	\epsilon=1   \longrightarrow {\omega_{sync}}_{ss} = 1, \quad
	\epsilon=0.1 \longrightarrow {\omega_{sync}}_{ss} = 0.1818.
\end{align*}
The convergency of $\omega_{sync}$ to 
${\omega_{sync}}_{ss}$ for $\epsilon=1$ is depicted in Fig.~\ref{fig:subplot}(b).

\begin{figure}[h]
\begin{center}
\includegraphics[width=0.5\textwidth]{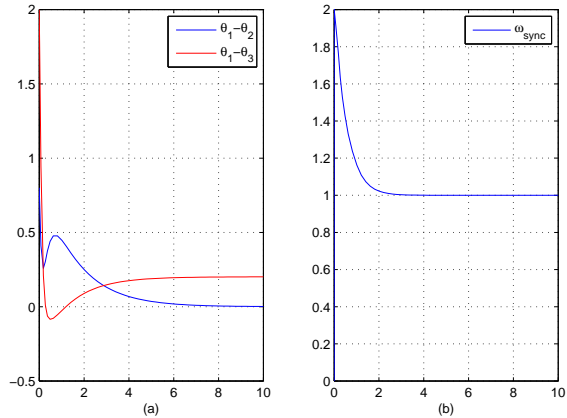}
\caption{(a) Line phases, 
         (b) convergence of $\omega_{sync}$ to  ${\omega_{sync}}_{ss}$}
\label{fig:subplot}
\end{center}
\end{figure}


\section{Conclusions}

We have analysed theoretical properties of inverter-based microgrids
controlled via primary and secondary loops. We have shown that
frequency regulation is ensured without
the need for time separation, and that ultimate boundedness of the
trajectories starting inside a region of the state space is guaranteed
under a condition on the inverters power injection errors. The
trajectory ultimate bound can be computed by simple iterations of a
nonlinear mapping and provides a certificate of the overall
performance of the controlled microgrid.  
Future work includes the derivation of design procedures based on 
the provided analysis, the extension of the results to more general controller
parameters and structures as well as relaxing some of the modelling
assumptions.

\bibliographystyle{plain}
\bibliography{DP_biblio}

\end{document}

%% file: texcommands_droop.tex
\newcommand{\R}{\mathbb{R}}

\newcommand{\Z}{\mathbb{Z}}

\newcommand{\I}{\mathcal{I}}
\newcommand{\J}{\mathcal{J}}

\newcommand{\diag}{\mathrm{diag}}

\newcommand{\bnd}{\mathbf{b}}
\newcommand{\mer}{\hfill$\circ$}

\newcommand{\1}{\mathbf{1}}
\newcommand{\0}{\mathbf{0}}

%% file: Droop_control_UB_Rahmat_v6.bbl
\begin{thebibliography}{10}

\bibitem{AiG13a}
N.~Ainsworth and S.~Grijalva.
\newblock {A structure-preserving model and sufficient condition for frequency
  synchronization of lossless droop inverter-based {AC} networks}.
\newblock {\em IEEE Trans.\ on Power Syst.}, 28(4):4310--4319, 2013.

\bibitem{AiG13}
N.~Ainsworth and S.~Grijalva.
\newblock Design and quasi-equilibrium analysis of a distributed
  frequency-restoration controller for inverter-based microgrids.
\newblock In {\em North American Power Symposium}, pages 1--6, 2013.

\bibitem{ASDJ12}
M.~Andreasson, H.~Sandberg, D.V. Dimarogonas, and K.H. Johansson.
\newblock Distributed integral action: Stability analysis and frequency control
  of power systems.
\newblock In {\em {IEEE} Conf.\ on Dec. and Control}, Hawai, 2012.

\bibitem{Bidram2012}
A.~Bidram and A.~Davoudi.
\newblock Hierarchical structure of microgrids control system.
\newblock {\em IEEE Transactions on Smart Grid}, 3(4):1963--1976, 2012.

\bibitem{CDA93}
M.C. Chandorkar, D.M. Divan, and R.~Adapa.
\newblock {Control of parallel connected inverters in standalone AC supply
  systems}.
\newblock {\em IEEE Transactions on Industry Applications}, 29(1):136--143,
  1993.

\bibitem{FD-JWSP-FB:14a}
F.~D{\"o}rfler, J.W. Simpson-Porco, and F.~Bullo.

\bibitem{GVMdVC11}
J.M. Guerrero, J.C. Vasquez, J.~Matas, L.G. {de Vicu\~na}, and M.~Castilla.
\newblock {Hierarchical Control of Droop-Controlled AC and DC Microgrids--A
  General Approach Toward Standardization}.
\newblock {\em IEEE Transactions on Industrial Electronics}, 58(1):158--172,
  2011.

\bibitem{haimovich12:_bound-arxiv}
H.~Haimovich and M.M. Seron.
\newblock Bounds and invariant sets for a class of switching systems with
  delayed-state-dependent perturbations, 2012.
\newblock {Available at http://arxiv.org/abs/1202.0455}.

\bibitem{HaS13}
H.~Haimovich and M.M. Seron.
\newblock Bounds and invariant sets for a class of switching systems with
  delayed-state-dependent perturbations.
\newblock {\em Automatica}, 49(3):748--754, 2013.

\bibitem{Kha02}
H.~Khalil.
\newblock {\em Nonlinear Systems}.
\newblock Prentice-Hall, NJ, 3rd edition, 2002.

\bibitem{Lasseter01}
R.H. Lasseter.
\newblock Microgrids.
\newblock In {\em IEEE Power Engineering Society Winter Meeting}, volume~1,
  pages 146--149, 2001.
\newblock Panel: Role of Distributed Generation in Reinforcing the Critical
  Electric Power.

\bibitem{PMM06}
J.A. {Pe\c{c}as Lopes}, C.L. Moreira, and A.G. Madureira.
\newblock {Defining control strategies for MicroGrids islanded operation}.
\newblock {\em IEEE Transactions on Power Systems}, 21(2):916--924, 2006.

\bibitem{PLT09}
F.Z. Peng, Y.W. Li, and L.M. Tolbert.
\newblock Control and protection of power electronics interfaced distributed
  generation systems in a customer-driven microgrid.
\newblock In {\em IEEE Power \& Energy Society General Meeting, PES'09}, 2009.

\bibitem{SATS12}
J.~Schiffer, A.~Anta, T.D. {Truong}, J.~Raisch, and T.~Sezi.
\newblock On power sharing and stability in autonomous inverter-based
  microgrids.
\newblock In {\em {IEEE} Conf.\ on Decision and Control}, Maui, Hawaii, USA,
  December 2012.

\bibitem{SVG12}
Q.~Shafiee, J.C. Vasquez, and J.M. Guerrero.
\newblock Distributed secondary control for islanded {MicroGrids}---a networked
  control systems approach.
\newblock In {\em Annual Conf.\ IEEE Industrial Electronics Society}, 2012.

\bibitem{JWSP-FD-FB:12u}
J.W. {Simpson-Porco}, F.~{D{\"o}rfler}, and F.~Bullo.
\newblock Synchronization and power sharing for droop-controlled inverters in
  islanded microgrids.
\newblock {\em Automatica}, 49(9):2603--2611, 2013.

\bibitem{Ustun2011}
T.S. Ustun, C.~Ozansoy, and A.~Zayegh.
\newblock {Recent developments in microgrids and example cases around the
  world. A review}.
\newblock {\em Renewable and Sustainable Energy Reviews}, 15(8):4030--4041,
  October 2011.

\bibitem{siljak2004}
A.I. Zecevic, G.Neskovic, and D.A. Siljak.
\newblock Robust decentralized exciter control with linear feedback.
\newblock {\em IEEE Transactions on Power Systems}, 19(2):1096--1103, 2004.

\bibitem{ZhH13}
Q.-C. Zhong and T.~Hornik.
\newblock {\em Control of Power Inverters in Renewable Energy and Smart Grid
  Integration}.
\newblock John Wiley \& Sons, 2013.

\end{thebibliography}
